\documentclass[runningheads]{llncs}
\usepackage{graphicx} 

\usepackage{hyperref}
\usepackage{amsmath}
\usepackage{amssymb}
\usepackage{amsfonts} 
\usepackage{xcolor}
\usepackage{graphicx}
\usepackage{comment}

\usepackage{todonotes}

\newcommand{\threefield}[3]{$#1\mid#2\mid#3$}

\title{Concurrency Constrained Scheduling with Tree-Like Constraints}
\author{Hans L. Bodlaender\inst{1} \and Danny Hermelin\inst{2} \and Erik Jan van Leeuwen\inst{1}}
\date{}

\institute{Department of Information and Computing Sciences, Utrecht University, Utrecht, the Netherlands \and Department of Industrial Engineering and Management, Ben-Gurion University of the Negev, Beer-Sheva, Israel}

\begin{document}

\sloppy
\maketitle

\begin{abstract}
This paper investigates concurrency-constrained scheduling problems, where the objective is to construct a schedule for a set of jobs subject to concurrency restrictions. Formally, we are given a conflict graph $G$ defined over a set of $n$ jobs, where an edge between two jobs in $G$ indicates that these jobs cannot be executed concurrently. Each job may have distinct attributes, such as processing time, due date, weight, and release time. The goal is to determine a schedule that optimizes a specified scheduling criterion while adhering to all concurrency constraints. This framework offers a versatile model for analyzing resource allocation problems where processes compete for shared resources, such as access to shared memory. From a theoretical perspective, it encompasses several classical graph coloring problems, including \textsc{Chromatic Number}, \textsc{Sum Coloring}, and \textsc{Interval Chromatic Number}.

\hspace{15pt}
Given that even the simplest concurrency-constrained scheduling problems are NP-hard for general conflict graphs, this study focuses on conflict graphs with bounded treewidth. Our results establish a dichotomy: Some problems in this setting can be solved in FPT time, while others are shown to be XALP-complete for treewidth as parameter. Along the way, we generalize several previously known results on coloring problems for bounded treewidth graphs. 
Several of the FPT algorithms are based on the insight that completion times are bounded by the Grundy number of the conflict graph --- the fact that this number is bounded by the product of treewidth and the logarithm of the number
of vertices then leads to the FPT time bound.

\keywords{treewidth \and
scheduling \and
parameterized complexity \and
concurrency constraints \and
FPT \and
XALP}

\end{abstract}

\section{Introduction}
\label{section:introduction}
In offline concurrency constrained scheduling, we have a set of $n$ jobs $\{1,\ldots,n\}$ for non-preemptive processing. We have an unlimited number ($n$ is enough) of identical machines. In the most basic setting, each job $j \in \{1,\ldots,n\}$ requires $p_j=1$ (unit) \emph{processing time} on any of the given machines. However, some pairs of jobs cannot be processed concurrently at the same time. These \emph{concurrency constraints} are given by a \emph{conflict graph} $G=(\{1,\ldots,n\},E)$, where two jobs $i$ and $j$ cannot be processed concurrently if and only if $\{i,j\} \in E$. Bar-Noy \emph{et al.}~\cite{Bar-NoyBHST98,Bar-NoyHKSS00} observed that concurrency constrained scheduling occurs in resource allocation problems with constraints imposed by the competition of processes on conflicting resource requirements. Similar situations occur in other settings, e.g.\ almost all modern database systems implement some sort of mutual exclusion scheme. A \emph{schedule} is a function $C: \{1,\ldots,n\} \to \mathbb{N}$ specifying the completion time of each job. We write $C_j=C(j)$ following standard scheduling notation. A schedule $C$ is said to be \emph{feasible} if $C(i) \neq C(j)$ for every $\{i,j\} \in E$. A feasible schedule thus is exactly a \emph{proper coloring} of $G$, where no two adjacent vertices get the same color. In concurrency constrained scheduling we aim to find a proper coloring of $G$ that optimises some scheduling criterion (as discussed below).



Following the three-field notation of Graham \emph{et al.} for scheduling problems~\cite{GrahamEtAl79}, we denote the concurrency constrained setting with unit processing times by \threefield{P_n}{conc,p_j=1}{\ast}. The first field denotes the machine model being used, in our case $n$ identical machines that run in parallel. The second field denotes constraints on the problem, where `conc' denotes that concurrency constrains (via a conflict graph) are present, while $p_j =1$ signifies that all jobs require a unit amount of processing time on any of the $n$ machines. Additional constraints and restrictions can be added to this field, as will be discussed later on. Finally, the last field `$*$' denotes the objective function that we optimise. 

Two notable and well-studied concurrency constrained scheduling problems are obtained by considering the \emph{makespan}~$C_{\max}= \max_j C_j$ and the \emph{total completion time} $\sum C_j$ objective functions. In the former, we wish to minimize the time in which the last job completes in our schedule. This directly translates to the classical \textsc{Chromatic Number} problem, as the minimum number of colors~$\chi(G)$ necessary to properly color $G$ is precisely $C_{\max}$ in a schedule with minimum makespan. In the latter, we wish to minimize the sum of completion times in our schedule, which is the same as minimizing the average completion time in the schedule. This directly translates to the \textsc{Chromatic Sum} problem~\cite{Kubicka04,KubickaS89,Supowit87}, a widely-studied graph coloring problem (see e.g.~\cite{Bar-NoyBHST98,Bar-NoyK98,BonomoV09,BorodinIYZ12,Gavril72,HalldorssonK02,Jansen97man,Jansen97,Kubicka04,KubickaS89,MalafiejskiGJK04,Marx05,NicolosoSS99,Salavatipour00,salavatipour2003sum,Szkaliczki99}) which has been originally studied in the context of VLSI design~\cite{Supowit87}. In our setting, these two problems are denoted by \threefield{P_n}{conc, p_j =1}{C_{\max}} and \threefield{P_n}{conc, p_j =1}{\sum C_j} respectively.

Besides unit processing times, we may also consider that
each job $j$ has an integer processing time $p_j \geq 1$ (which it requires on any machine). The time intervals in which two conflicting jobs $i$ and $j$ are processed must be disjoint. Thus, we now require a feasible schedule to have:
\begin{itemize}
\item $C_j - p_j \geq 0$ for all $j \in \{1,\ldots,n\}$ (no jobs can start before time 0), \emph{and}
\item $[C_i-p_i,C_i) \cap [C_j-p_j,C_j) = \emptyset$ for all $\{i,j\} \in E$ (all concurrency constraints are met).
\end{itemize}
We denote this setting by 
\threefield{P_n}{conc}{\ast}. Note that in the case of arbitrary processing times, some proper colorings of $G$ no longer correspond to feasible schedules. However, every feasible schedule is still a proper coloring. 

The \threefield{P_n}{conc}{C_{\max}} problem was previously studied in the literature under the names \textsc{Interval Chromatic Number}~\cite{Golumbic80,Kubale89} or \textsc{Non-Preemptive Multichromatic Number}~\cite{Bar-NoyHKSS00}. When restricting the conflict graph $G$ to be an interval graph we get the classical \textsc{Dynamic Storage Allocation} problem~\cite{GareyJ79} (also known as the \textsc{Ship Building} problem~\cite{Golumbic80}), and when $G$ is a line graph we get the most basic case of the \textsc{Scheduling File Transfers} problem~\cite{CoffmanEtAl85}. The \threefield{P_n}{conc}{\sum C_j} problem was previously studied in the literature under the name \textsc{Non-Preemptive Multichromatic Sum}~\cite{Bar-NoyHKSS00}.

Apart from $C_{\max}$ and $\sum C_j$, there are several other natural criteria in the context of concurrency constrained scheduling. First, we can assume that each job~$j$ has a due date~$d_j$ associated with it, where if it completes after time~$d_j$ there is a penalty incurred to our schedule. Two such criteria that are widely considered in scheduling literature (see \emph{e.g.}~\cite{PinedoBook})  are the \emph{maximum lateness} criterion and the \emph{total tardiness} criterion. In the former, we wish to minimize the maximum lateness $L_{\max}= \max_j L_j$ in our schedule, where the \emph{lateness} of a job~$j$ is defined by~$L_j= C_j-d_j$. In the latter, we wish to minimize the total tardiness $\sum_j T_j$ of our schedule, where the \emph{tardiness} of a job~$j$ is given by~$T_j=\max\{0, C_j-d_j\}$. Note that $C_{\max}$ is a special case of $L_{\max}$ where each job has due date $d_j=0$, and in this special case $\sum T_j$ becomes $\sum C_j$ as well. Weighted variants, where each job $j$ has an additional integer weight~$w_j$, exist for  all criteria above as well.

Second, we can consider a release time $r_j \geq 0$ for each job $j \in \{1,\ldots,n\}$. Note that in the above, all jobs are released at time $r_j=0$. Release times are different than due dates in that release times are hard constraints that must be satisfied, where we require any feasible schedule to have $C_j - p_j \geq r_j$ for all jobs~$j$. We assume throughout the paper that all job parameters $p_j$, $d_j$, $w_j$, and~$r_j$ are given in unary, and so are bounded by the total input length.

\paragraph{Bounded treewidth conflict graphs:}

As the most basic cases of concurrency constrained scheduling, \textsc{Chromatic Number} and \textsc{Chromatic Sum}, are already NP-hard~\cite{Karp72,KubickaS89}, it makes sense to study cases where the conflict graph has some restricted structure. We focus on graphs of bounded treewidth and pathwidth. 
\textsc{Chromatic Number} is well known to be solvable in FPT time parameterized by treewidth, as there exists a $O^*(q^{tw+1})$-time\footnote{Here and elsewhere in the paper we use $O^*()$ to suppress polynomial factors.} algorithm for determining whether a graph~$G$ can be properly colored with $q$ colors given a tree decomposition of width $tw$ for $G$~\cite{CyganFKLMPPS15,TelleP1997}. 
Moreover, this running time cannot substantially be improved assuming the Strong Exponential Time Hypothesis (SETH)~\cite{LokshtanovMS18}. 
Since the chromatic number of a graph is at most its treewidth, this yields an algorithm with running time $O^*(tw^{tw+1} )$ for \threefield{P_n}{conc, p_j =1}{C_{\max}}. 
The \textsc{Chromatic Sum} (\emph{i.e.} \threefield{P_n}{conc, p_j =1}{\sum C_j}) problem has also been previously studied in the context of bounded treewidth graphs.
Specifically, Jansen~\cite{Jansen97} proved that an optimal coloring for bounded treewidth graphs uses $O(tw \log n)$ colors, and used the $O^*(q^{tw+1})$-time algorithm mentioned above to solve a slightly more generalised version of the \threefield{P_n}{conc, p_j =1}{\sum C_j} problem in $O^*((tw \cdot \log n)^{tw+1})$ time (which is also FPT time, albeit a bit slower than $O^*(tw^{tw+1})$). 
Halld\'{o}rsson and Kortsarz showed that the more general \threefield{P_n}{conc, p_j\leq p}{C_{\max}} and \threefield{P_n}{conc, p_j\leq p}{\sum C_j} problems, where all processing times are bounded by some fixed integer $p$, can be solved in $O^*((p \cdot tw \cdot \log n)^{tw+1})$ time~\cite{HalldorssonK02}.

\paragraph{XNLP/XALP problems:}
XNLP, introduced by~\cite{BODLAENDER2024105195,ElberfeldST15}, is the class of all problems solvable in FPT non-deterministic time using only $f(k) \log n$ space. Several natural problems problems parameterized by pathwidth turned out to be complete for this class~\cite{BodlaenderGJPP22}, and in particular problems that require storing $n^{\Omega(f(pw))}$ space (for a non-decreasing $f$ that grows to infinity) in their standard path decomposition dynamic programming algorithm. The class XALP was introduced by~\cite{BodlaenderGJPP22} as analogue to XNLP. Similar to XNLP, XALP (formally defined in Section~\ref{section:preliminaries}) captures problems whose dynamic programming algorithms on tree decompositions require $n^{\Omega(f(tw))}$ time and space for similar $f$ (see e.g.,~\cite{BodlaenderGJPP22,BodlaenderMOPL23,BodlaenderS24}).


\subsection*{Our results}
We investigate in detail the complexity of concurrency-constrained scheduling problems parameterized by the treewidth (resp.~pathwidth) of the conflict graph. We provide a complete classification of all problems discussed above into two categories: those that are solvable in FPT time and those that are XALP-complete (resp.~XNLP-complete). See the overview in Table~\ref{tab:unweighted}.

\begin{table}[bt]
\begin{center}
\begin{tabular}{c|c|c||c|c}
 & $C_{\max}$ & $\sum C_j$ & $L_{\max}$ & $\sum T_j$ \\
\hline
\hline
$r_j =0, p_j = 1$ & $tw^{O(tw)}$~\cite{TelleP1997} & $(tw \cdot \log n)^{O(tw)}$~\cite{Jansen97} & $(tw \cdot \log n )^{O(tw)}$ & $(tw \cdot \log n )^{O(tw)}$\\
 & $pw^{O(pw)}$~\cite{TelleP1997} & $pw^{O(pw)}$ \cite{Jansen97} & $pw^{O(pw)}$ & $pw^{O(pw)}$\\
\hline
\hline
$r_j \geq 0, p_j = 1$ & $(tw \cdot \log n)^{O(tw)}$ & XALP-complete & XALP-complete & XALP-complete\\
& $pw^{O(pw)}$ & XNLP-complete & XNLP-complete & XNLP-complete \\
\hline
\hline
$r_j =0, p_j \geq 1$ & XALP-complete & XALP-complete & XALP-complete & XALP-complete\\
&  XNLP-complete & XNLP-complete & XNLP-complete & XNLP-complete \\
\end{tabular}
\end{center}
\caption[]{Complexity results for unweighted concurrency constrained scheduling problems, with respectively treewidth and pathwidth as parameter. The columns represent the objective function, while the rows correspond to constraints on release times and processing times. The running times presented ignore polynomial factors in $n$. }
\label{tab:unweighted}
\end{table}


On the positive side, we show that for the case of unit processing times and no release times (\emph{i.e.} $p_j=1$ and $r_j=0$ for all $j \in \{1,\ldots,n\}$), all problems we consider are FPT. In particular, we generalize the result of Jansen~\cite{Jansen97} and show that for any so called \emph{regular scheduling criterion}~\cite{PinedoBook}, \emph{i.e.} any scheduling criterion which is non-decreasing in the completion time of the jobs, there exists an optimal schedule $C$ with makespan $C_{\max} \leq tw(G) \cdot \log n +1$. All scheduling criteria discussed above, as well as many others, are regular. We also show that this bound increases by a factor of $O(p)$ for the case where the processing time of each job is at most $p$, generalizing the result of Halld\'{o}rsson and Kortsarz~\cite{HalldorssonK02}. 

Some of our FPT algorithms are based on the insight that completion times are bounded by the Grundy number of the conflict graph. Telle and Proskurowski~\cite{TelleP1997} showed that the Grundy number of a graph with treewidth $tw$ is $O(tw \cdot \log n)$. If $k$ bounds the maximum completion time, then an algorithm with running time $k^{O(tw)} \cdot n^{O(1)}$ and the Grundy number bound implies the problem is FPT, as $O(tw \cdot \log n)^{O(tw)}\cdot n^{O(1)}$ is of the form $f(tw)\cdot n^{O(1)}$ (cf.~\cite[p.~74, Answer 3.18]{CyganFKLMPPS15}).

On the negative side, we show that if jobs have unbounded processing times or release times,
then only
\threefield{P_n}{conc,r_j \geq 0,p_j=1}{C_{\max}}
remains FPT;  all other problems become XALP-complete. 
Our hardness results here are obtained by developing reductions from the XALP-complete \textsc{Precoloring Extension} problem. 
We note that XALP-completeness results have two important implications.
First, these problems are in XP,
i.e., have an algorithm with running time of the form $O(|I|^{f(k)})$, where $|I|$ is the size of the instance.
Second, a conjecture by Pilipczuk and Wrochna~\cite{PilipczukW18} implies that XALP-hard problems are unlikely to have an XP-algorithm with space usage $f(k) \cdot |I|^{O(1)}$.

We also consider the weighted variants of our problems, where each job $j$ has an integer weight $w_j$ specifying its significance. In this context, we write $WC_{\max}= \max_j w_jC_j$ and $WL_{\max}= \max_j w_jL_j$. Again, when $p_j=1$ and $r_j=0$ for all jobs $j$, all problems are FPT, although with a slightly slower running time of $O^*((tw \cdot \log n )^{tw+1}))$ for \threefield{P_n}{conc,p_j=1}{WC_{\max}} in comparison to the unweighted case. We also show that \threefield{P_n}{conc,p_j=1,r_j \geq 0}{WC_{\max}} is XALP-complete, as opposed to the unweighted \threefield{P_n}{conc,p_j=1,r_j \geq 0}{C_{\max}} problem which is FPT. This hardness result again follows via a reduction from \textsc{Precoloring Extension} parameterized by treewidth.

For pathwidth, the running time of the algorithms decreases, as the Grundy number of a graph of pathwidth $pw$ is~$O(pw)$~\cite{BelmonteKLMO22,DujmovicJW12}. For the hardness results, since \textsc{Precoloring Extension} is XNLP-complete parameterized by pathwidth~\cite{BODLAENDER2024105195}, we note that we can reuse the above hardness constructions.

\section{Preliminaries}
\label{section:preliminaries}
We use $G=(\{1,\ldots,n\},E)$ to denote the conflict graph of our concurrency constrained scheduling problems. Standard graph-theoretic notation is found in~\cite{Diestel}. We assume readers are familiar with the concepts of treewidth and pathwidth.

\paragraph{XALP-completeness:}

The class XALP is the class of all parameterized problems that can be solved in $f(k) \cdot n^{O(1)}$ time and $f(k) \cdot \log n$ space on a non-deterministic Turing Machine with access to an auxiliary stack (with only top element lookup allowed)~\cite{BodlaenderGJPP22}. XALP is closed by reductions using at most $f(k) \log n$ space and FPT time. These conditions are implied by using at most $f(k) + O(\log n)$ space. Reductions
respecting the latter condition are \emph{parameterized logspace} reductions.

We use the \textsc{Precoloring Extension} problem frequently in our reductions. In \textsc{Precoloring Extension}, we are given a graph $H=(V,F)$, and integer $k$, and a proper coloring $\chi:V_0 \to \{1,\ldots,k\}$ defined only on a subset of vertices $V_0 \subseteq V$. Our goal is to determine whether we can extend $\chi$ to the set $V \setminus V_0$ of remaining vertices. That is, is there a proper coloring $\chi' : V \to \{1,\ldots,k\}$ with $\chi'(v)=\chi(v)$ for all $v \in V_0$. It is XALP-complete for parameter treewidth~\cite{BodlaenderGJPP22} and XNLP-complete for parameter pathwidth \cite{BODLAENDER2024105195}.

\paragraph{Tree decomposition dynamic programming:}
We show that one can compute an optimal schedule~$C$ for any concurrency constrained scheduling problem considered in this paper in $O^*(k^{tw+1})$ time, assuming there exists such a schedule with $C_{\max} \leq k$. We prove this specifically for the two most general problems we consider in the paper, the \threefield{P_n}{conc,r_j \geq 0}{WL_{\max}} and \threefield{P_n}{conc,r_j \geq 0}{\sum w_jT_j} problems. The proof uses an algorithm very similar to the $O^*(k^{tw+1})$-time algorithm for properly coloring graphs of treewidth $tw$ with $k$ colors~\cite{CyganFKLMPPS15,TelleP1997}.

\begin{theorem}
\label{thm:TreewidthDP}%
The \threefield{P_n}{conc,r_j \geq 0}{WL_{\max}} and \threefield{P_n}{conc,r_j \geq 0}{\sum w_jT_j} problems are both solvable in $O^*(k^{tw+1})$-time, where $k$ is a bound on the makespan of some optimal schedule.
\end{theorem}

\begin{proof}
Let $G=(\{1,\ldots,n\},E)$ be a conflict graph with $tw(G)=tw$. Using Korhonens approximation algorithm for tree decompositions~\cite{Korhonen21}, we can compute in $2^{O(tw)} \cdot n$ time a tree decomposition of width $O(tw)$ and $O(tw \cdot n)$ nodes for $G$. We then convert the tree decomposition into a \emph{nice tree decomposition} of equal width and a linear number of nodes~\cite{Kloks93}. A nice tree decomposition $\mathcal{T}=(\mathcal{X},F)$ of $G$ is a rooted tree decomposition where we have four types of nodes: 
\begin{itemize}
\item \emph{Leaf nodes} $X \in \mathcal{X}$ which have no children in $\mathcal{T}$, and are singletons. That is, $X=\{j\}$ for some vertex $j \in \{1,\ldots,n\}$.
\item \emph{Forget nodes} $X \in \mathcal{X}$ which have a single child $X'$ in $\mathcal{T}$ such that $X=X' \setminus \{j\}$ for some $j \in X'$.
\item \emph{Introduce nodes} $X \in \mathcal{X}$ which have a single child $X'$ in $\mathcal{T}$ such that $X=X' \cup \{j\}$ for some $j \in \{1,\ldots.n\} \setminus X'$.
\item \emph{Join nodes} $X \in \mathcal{X}$ which have two children $X'$ and $X''$ in $\mathcal{T}$ such that $X=X'=X''$.
\end{itemize}

According to the above, we proceed by assuming we have a nice tree decomposition $\mathcal{T}=(\mathcal{X},F)$ for $G$. For a node $X$ of $\mathcal{T}$, let $G[X]$ denote the subgraph of $G$ induced by $X$, and let $G[\mathcal{T}(X)]$ denote the subgraph of $G$ induced by $X$ and all of its descendants in $\mathcal{T}$. The algorithms we describe both compute dynamic program tables $DP_X$ for each node $X$ of $\mathcal{T}$ in bottom-up fashion, which have an entry for each possible schedule (coloring) $C: X \to \{1,\ldots,q\}$ of $X$ with makespan at most $q$. For a given schedule $C$, the entry $DP_X[C]$ will equal the value (which differs according to whether $WL_{\max}$ or $\sum w_j T_j$ is the objective function) of an optimal schedule $C'$ for all jobs in $G[\mathcal{T}(X)]$ such that $C'(j)=C(j)$ for all $j \in X$. If there is no such feasible schedule $C'$, we set $DP_X[C] = \infty$.  

The dynamic programming computation is done in bottom-up fashion on $\mathcal{T}$ as follows:
\begin{itemize}
\item Suppose $X$ is a leaf node with $X=\{j\}$ for some job $j \in \{1,\ldots,n\}$. Then for any schedule $C: \{j\} \to \{1,\ldots,k\}$ with $C(j) < r_j+ p_j$, we set $DP_X[C]=\infty$. For a schedule $C$ with $C(j) \geq r_j+p_j$ we set $DP_X[C]=w_j \cdot (C(j) - d_j)$ in the case of the $WL_{\max}$ objective, and $DP_X[C]=w_j \cdot \max\{0,C(j) - d_j\}$ in the case of the $\sum w_j T_j$ objective.
\item Suppose $X$ is forget node with child $X'$, where $X=X' \setminus \{j\}$ for some $j \in X'$. Then for any schedule $C: X \to \{1,\ldots,k\}$, we set $DP_X[C] = \min_{C'} DP_{X'}[C']$, where the minimum is taken over all extensions~$C'$ of $C$ onto $X'$. 
\item Suppose $X$ is an introduce node with child $X'$, where $X=X' \cup \{j\}$ for some $j \in \{1,\ldots,n\} \setminus X'$. Let $C: X \to \{1,\ldots,k\}$ be any schedule for $X$, and let $C'$ be the restriction of $C$ onto $X'$. If $C(j) < r_j + p_j$ we set $DP_X[C]=\infty$. Otherwise, we set $DP_X[C] = \max\{w_j \cdot (C(j)-d_j),DP_{X'}[C']\}$ in the case of the $WL_{\max}$ objective, and $DP_X[C]=w_j \cdot \max\{0,C(j) - d_j\} + DP_{X'}[C']$ in the case of the $\sum w_j T_j$ objective.
\item Suppose $X$ is a join node with children $X'$ and $X''$, where $X=X'=X''$. Then for any schedule $C: X \to \{1,\ldots,k\}$, we set $DP_X[C] = \max\{DP_{X'}[C],DP_{X''}[C]\}$ in the case of the $WL_{\max}$ objective, and $DP_X[C]=DP_{X'}[C]+DP_{X''}[C] - \sum_{j \in X} w_j \cdot \max\{0,C(j) - d_j\}$ in the case of the $\sum w_j T_j$ objective.
\end{itemize}

For both objective functions, the value of the optimal solution is given by taking the minimum entry~$T_X[C]$ for the root node $X$ of $\mathcal{T}$. Correctness of this algorithm follows similar arguments used for proving the correctness of the $k$-coloring algorithm for bounded treewidth graphs~\cite{CyganFKLMPPS15,TelleP1997}. Since computing any entry $DP_X[C]$ can be done in $O(k)$ time, and there are $O^*(k^{tw+1})$ entries in total, the running time of the theorem follows as well. \qed
\end{proof}

Observe that Theorem~\ref{thm:TreewidthDP} above implies that both \threefield{P_n}{conc,r_j \geq 0}{WL_{\max}} and \threefield{P_n}{conc,r_j \geq 0}{\sum w_jT_j} are solvable in $|I|^{O(tw)}$-time, where $|I|$ denotes the size of the input, as we assume that all processing times are given in unary. Hence, the total processing time of all jobs, and the makespan of any schedule, is always polynomial in~$|I|$. 
The dynamic programming algorithm of Theorem~\ref{thm:TreewidthDP} can even be turned into an XALP-membership proof:


\begin{corollary}
\threefield{P_n}{conc,r_j \geq 0}{WL_{\max}} and \threefield{P_n}{conc,r_j \geq 0}{\sum w_jT_j} are both in XALP with treewidth as parameter.     
\label{corollary:inXALP}
\end{corollary}
\begin{proof}
We can turn the dynamic programming algorithm from Theorem~\ref{thm:TreewidthDP} into an XALP-membership proof. The argument is similar to XALP-membership proofs in Bodlaender \emph{et al.}~\cite{BodlaenderGJPP22}. Instead of building all tables, we non-deterministically guess an element from a table, working bottom-up. This can be done by performing a postorder traversal of the tree. After handling a node $X$, we push a `guessed' coloring $C$ of $X$  on the stack, along with a computed value~$DP_X[C]$. When handling an introduce or forget node, we pop one element from the stack; and when handling a join node, we pop two elements from the stack. In a leaf node, we guess $C$ and compute $DP_X[C]$ as given. In an introduce node, we obtain from the stack a pair consisting of~$C'$ and a value $DP_X[C']$, then guess~$C(j)$ non-deterministically, and use the given formula to obtain~$DP_X[C']$. In a forget node, we obtain from the stack the pair $C'$ with a value $DP_X[C']$, and set the new pair to be the projection to $X$. The computation in a join node follows the given formulas in the proofs above. Note that if we correctly guess all values~$C(j)$ in some optimal solution, then this computation gives the correct answer. At each point, we use $f(tw) \cdot \log n$ space outside the stack: $O(\log n)$ bits for $O(1)$ positions in the tree, and at most three pairs $(C,DP_X[C])$.  
\qed
\end{proof}


\section{Bounded Processing Times}
\label{section:boundedprocessingtimes}

In the following section we focus on cases where job processing times are either unit, or relatively small, \emph{i.e.} \threefield{P_n}{conc, p_j=1}{\ast} and \threefield{P_n}{conc, p_j \leq p}{\ast} for some fixed $p \in O(1)$. Previous work has shown that for the $\sum C_j$ objective, there always exists an optimal schedule $C$ with makespan at most $tw(G) \cdot \log n$ for the unit processing time case~\cite{Jansen97}, and $tw(G) \cdot p \log n$ for the $p_j \leq p$ case~\cite{HalldorssonK02}. In what follows we generalize and unify these results into our framework. For this, we first introduce the notion of minimal feasible schedules.  

\begin{definition}
A feasible schedule $C$ is said to be \emph{minimal} for a set of jobs $\{1,\ldots,n\}$ if decreasing the completion time of any job by any amount of time causes $C$ to become infeasible.     
\end{definition}


We begin by considering the case of unit processing times. Note that this case is equivalent to the case of $p_j=p$ for all jobs $j$, since one can simply multiply the completion times in any feasible schedule for a \threefield{P_n}{conc, p_j=1}{\ast} instance by~$p$. We will heavily rely on the notion of a \emph{Grundy number} of a graph: A \emph{greedy coloring} of a graph $G=(V,E)$ is a proper coloring $C: \{1,\ldots,n\} \rightarrow \mathbb{N}$ such that for every vertex~$j$ and any color $c \in \{1,\ldots,C(j)-1\}$ there is some vertex $i$ with $\{i,j\} \in E$ and $C(i)=c$. The \emph{value} of a greedy coloring $C$ of $G=(V,E)$ is $C_{\max}=\max_{v\in V} C(v)$. The \emph{Grundy number} of $\Gamma(G)$ of a graph~$G$ is the maximum value of any greedy coloring of $G$. From \cite{TelleP1997}, we have:
\begin{lemma}[\cite{TelleP1997}]
\label{lem:Grundy}
$\Gamma(G) \leq tw(G) \cdot \log n + 1$ for any graph $G$ on $n$ vertices.    
\end{lemma}

\begin{lemma}
\label{lem:MinimalSchedulesAreGreedy}%
A feasible minimal schedule of any \threefield{P_n}{conc, p_j=1}{\ast} instance corresponds to a greedy coloring of the conflict graph $G$. 
\end{lemma}

\begin{proof}
Consider a feasible minimal schedule for the set of jobs $\{1,\ldots,n\}$ given in some \threefield{P_n}{conc, p_j=1}{\ast} instance, and let $j \in \{1,\ldots,n\}$. We argue that for every $c \in \{1,\ldots,C(j)-1\}$ there exists some $i \in \{1,\ldots,n\}$ with $\{i,j\} \in E$ and $C(i)=c$. Suppose not. That is, suppose there is some $c \in \{1,\ldots,C(j)-1\}$ such that for all $i \in \{1,\ldots,n\}$ with $\{i,j\} \in E$ we have $C(i)\not=c$. Then, since all processing times are unit, the schedule $C'$ with $C'(j)=c$ and $C'(i)=C(i)$ for all $i \in \{1,\ldots,n\} \setminus \{j\}$ is feasible as well, contradicting our assumption that $C$ is minimal. It follows that $C$ is a greedy coloring of $G$. \qed
\end{proof}

Recall that a scheduling criterion is \emph{regular} if it is non-decreasing in the completion time of the jobs; i.e., decreasing the completion time of any job does not increase the value of the schedule. We thus combine Lemma~\ref{lem:Grundy} and~\ref{lem:MinimalSchedulesAreGreedy} to get:
\begin{corollary}
\label{cor:MinimalUnitProc}%
Any \threefield{P_n}{conc, p_j=1}{\ast} instance where $*$ is regular has an optimal schedule $C$ with $C_{\max} \leq tw \cdot \log n + 1$.
\end{corollary}

We now discuss the tightness or non-tightness of the bound of Corollary~\ref{cor:MinimalUnitProc} for the different objective functions.

The bound of the Corollary~\ref{cor:MinimalUnitProc} is clearly not tight for the makespan objective, as there we have $C_{\max} \leq tw(G)+1$: we can vertex color $G$ with $tw(G)+1$
colors, and schedule each job colored $i$ at time $i$.
However, for $\sum C_j$ it is known to be almost tight in the sense that there are instances of \threefield{P_n}{conc, p_j=1}{\sum C_j} where $G$ is a tree (\emph{i.e.} $tw(G)=1$) and any optimal schedule $C$ has $C_{\max} = \Omega(\log n)$ \cite{KubickaS89}. As $\sum C_j$ is a special case of $\sum T_j$, $\sum T_j$, and $\sum w_jT_j$, this also applies to these objective functions as well. In the two lemmas below we show an even stronger bound for the remaining $WC_{\max}$ and $L_{\max}$ objectives.  

\begin{lemma}
\label{lem:treeWCmax}%
There exists a \threefield{P_n}{conc, p_j=1}{WC_{\max}} instance where $tw(G)=1$ and any optimal schedule has $C_{\max} = \log n + 1$.
\end{lemma}

\begin{proof}
Let $W \geq \log n + 1$. We construct a sequence of rooted conflict trees $T_1, T_2,\ldots$ as follows: $T_1$ has a single node (job) with weight $W$ which is also its root. For $i > 1$, tree~$T_i$ is obtained by creating a new root job with weight 1, and connecting this job to the roots of trees $T_1,\ldots,T_{i-1}$. Let $|T_i|$ denote the number of jobs in $T_i$. Then $|T_i|=|T_1|+\cdots+|T_{i-1}|+1=2^{i-1}$. Moreover, it is not difficult to show by induction on~$i$ that the completion time of the root of $T_i$ is exactly $i$ in any schedule with $WC_{\max} \leq W$ (\emph{i.e.} any schedule where vertices of weight $W$ complete at time 1). The proof now follows by taking the \threefield{P_n}{conc, p_j=1}{WC_{\max}} instance corresponding to tree~$T_{\log n+1}$.
\qed
\end{proof}

\begin{lemma}
There exists a \threefield{P_n}{conc, p_j=1}{L_{\max}} instance where $tw(G)=1$ and any optimal schedule has $C_{\max} = \log n + 1$.
\end{lemma}

\begin{proof}
The proof is similar to the proof of Lemma~\ref{lem:treeWCmax}. We construct a sequence of rooted conflict trees as follows: $T_1$ has a single job with due date 1 which is also its root. For $i > 1$, the tree~$T_i$ is obtained by creating a new root job with due date~$i$, and connecting this job to the roots of trees $T_1,\ldots,T_{i-1}$. Again we have that $|T_i|=2^{i-1}$, and any optimal schedule (with $L_{\max}=0$) must schedule  the root of $T_i$ to complete at time~$i$. The proof then follows by taking the \threefield{P_n}{conc, p_j=1}{L_{\max}} instance corresponding to tree~$T_{\log n+1}$.
\qed
\end{proof}

Note that all scheduling criteria discussed in this paper are regular. Thus, by using the algorithm of Theorem~\ref{thm:TreewidthDP}, we get an algorithm for all these criteria in the unit processing-times case. We state this for the two most general criteria:
\begin{theorem}
\label{thm:UnitProc}%
The \threefield{P_n}{conc,p_j=1}{WL_{\max}} and \threefield{P_n}{conc,p_j=1}{\sum w_jT_j} problems are both solvable in $O^*((tw \cdot \log n+1)^{tw+1})$-time.    
\end{theorem}


\paragraph{Varying processing times}

We next consider the case where the processing times of jobs may vary, but are still rather small. We denote this case by \threefield{P_n}{conc, p_j \leq p}{\ast}. We extend the bound of Corollary~\ref{cor:MinimalUnitProc} to this case as well, exploiting the connection between treedepth and treewidth, although at the price of being slightly less precise.

For the proof of Lemma~\ref{lem:3} below, we will need to first define the notion of \emph{treedepth} of a graph. The treedepth $td(G)$ of a (connected) graph $G=(\{1,\ldots,n\},E)$ is the minimum height of any rooted tree $T$ over $\{1,\ldots,n\}$ with the property that for every $\{i,j\} \in E$ we have that $i$ and $j$ have an ancestor-descendant relationship to each other in $T$. We will use the following known relationship between the treewidth and the treedepth parameters of a graph:

\begin{lemma}[\cite{BodlaenderGHK1995}]
\label{lem:TreedepthTreewdith}
$td(G) = O(tw(G) \cdot \log n)$ for any graph $G$ on $n$ vertices.    
\end{lemma}

\begin{lemma}
There exists a feasible minimal schedule $C$ for any \threefield{P_n}{conc, p_j\leq p}{\ast} instance such that $C_{\max} =O(tw(G) \cdot p \log n)$.
\label{lem:3}
\end{lemma}

\begin{proof}
Consider a rooted tree $T$ over $\{1,\ldots,n\}$ of height $height(T)=td(G)$ such that $i$ and $j$ have an ancestor-descendant relationship in $T$ for every $\{i,j\} \in E$. For $i=1,2,\ldots$, consider the following procedure: Let $J_i$ be the set of jobs~$j$ with $C_j \in [(i-1)p+1, \ldots, ip]$. Remove the set~$J_i$ from the current set of jobs, and set $i=i+1$. We claim that there are no more jobs remaining after $td(G)$ iterations of this procedure, implying that $C_{\max} \leq (td(G) + 1)p$.

Consider the first iteration $i=1$. Then for any leaf $\ell$ of $T$, there is at least one job of $J_1$ in the unique root-to-$\ell$ path $P_\ell$ of $T$. Otherwise, we can safely set $C_\ell= p_\ell$, since all jobs conflicting with $\ell$ are on $P_\ell$, and none of these are processed within the time interval $[0 \ldots, p) \supseteq [0, \ldots, C_\ell)$. This leads to a contradiction, as we assume $C$ is minimal. It follows that every $P_\ell$ path in $T$ contains at least one job of $J_i$. Now consider the tree $T_1$ obtained be removing $J_1$ from $T$ without changing the ancestor-descendant relationship of the remaining jobs. Then $height(T_1)=height(T)-1$. The same argument also holds for iterations $i > 1$: If for some leaf $\ell$  of $T_i$ the path $P_\ell$ does not contain any job of $J_i$, we can safely set $C_\ell=(i-1)p+p_{\ell}$, contradicting the minimality of $C$. Thus, $height(T_i)=height(T_{i-1})-1$ for all $i$, and so the process above terminates after $td(G)$ iterations. This implies that~$C_{\max} \leq (td(G) + 1)p$, and so combined with Lemma~\ref{lem:TreedepthTreewdith}, we get that $C_{\max} =O(tw(G) \cdot p \log n)$.
\qed
\end{proof}

\begin{theorem}
\label{thm:VaryingProc}%
The \threefield{P_n}{conc,p_j \leq p}{WL_{\max}} and \threefield{P_n}{conc,p_j \leq p}{\sum w_jT_j} problems are both solvable in $O^*((tw \cdot p \cdot \log n)^{O(tw)})$-time.    
\end{theorem}


\section{Unbounded Processing Times}
\label{section:unboundedprocessingtimes}

We next consider the case where the processing times of jobs are unbounded, \emph{i.e.} problems of the form \threefield{P_n}{conc}{\ast}. We show that already for the two most basic objective functions of $C_{\max}$ and $\sum C_j$, the corresponding scheduling with concurrency constrained problem becomes XALP-complete parameterized by the treewidth of the conflict graph. As \threefield{P_n}{conc}{C_{\max}} and \threefield{P_n}{conc}{\sum C_j} are special cases of any \threefield{P_n}{conc}{\ast} problem considered in this paper, this will complete the picture for the second row in Table~\ref{tab:unweighted}. For both these problems, we show XALP-hardness by a reduction from \textsc{Precoloring Extension}. 

\subsection{The Makespan objective}

Let $(H,k,\chi)$ be a given instance of \textsc{Precoloring Extension} with $H=(V,F)$, $V=\{1,\ldots,n\}$, $k$ an integer, and $\chi: V_0 \to \{1,\ldots,k\}$ a given precoloring of $H$ defined on a subset of vertices $V_0 \subseteq V$. 
We may assume that $H$ has at least one edge. 
The main idea of our reduction is as follows: We construct an instance of \threefield{P_n}{conc}{C_{\max}} by initially starting with the set of $n$ jobs $\{1,\ldots,n\}$ corresponding to the vertices of~$H$, each with processing time $p_j=1$, and a conflict graph~$G$ which is identical to $H$. We then add auxiliary jobs to each job $j \in V_0$ whose goal is to ensure that we get~$C_j=\chi(j)$ in any feasible schedule $C$ with  makespan at most $k$. In this way, feasible schedules of makespan at most $k$ correspond to valid extensions of $\chi$ when restricted to $\{1,\ldots,n\}$.

Our construction proceeds as follows: We start with $G=H$, and define the processing time of each job $j \in \{1,\ldots,n\}$ to be $p_j=1$. We then construct four jobs $a$, $a'$, $b$, and $b'$ such that:
\begin{itemize}
\item $p_a = p_b =1$ and $p_{a'}=p_{b'}=k-1$. 
\item $\{a,a'\}$, $\{b,b'\}$, and $\{a,b\}$ are edges in $G$.
\end{itemize}
The purpose of $a$ and $b$ is to serve as delimiter jobs, as is formalized in the following easy to verify lemma:
\begin{lemma}
\label{lem:delimiters}%
If $C$ is a feasible schedule for $\{a,a',b,b'\}$ with $C_{\max} \leq k$ then $\{C_a,C_b\}=\{1,k\}$.   
\end{lemma}

Next add at most two additional jobs per each precolored $j \in V_0$. Let $j \in V_0$, and suppose that $\chi(j) \in \{2,\ldots,k-1\}$. Create two jobs $j(1)$ and $j(2)$ such that:
\begin{itemize}
\item $p_{j(1)} = \chi(j)-1$ and $p_{j(2)} = k-\chi(j)$.
\item $\{j,j(1)\}$, $\{j,j(2)\}$, and $\{j(1),j(2)\}$ are edges in $G$.
\item $\{a,j\}$, $\{b,j\}$ $\{a,j(2)\}$ and $\{b,j(1)\}$ are edges in $G$.
\end{itemize}
For $j \in V_0$ with $\chi(j)=1$, we create only job $j(2)$ above, and we do not add the edge $\{a,j\}$ to $G$. For $j \in V_0$ with $\chi(j)=k$, we create only job~$j(1)$ above,  and we do not add the edge $\{b,j\}$ to $G$. This completes our construction. 

\begin{lemma}
\label{lem:TreewidthCmax}
$tw(G)\leq tw(H)+2$ and $pw(G) \leq pw(H)+4$.
\end{lemma}
\begin{proof}
Let $\mathcal{T}_H= (\mathcal{X},F)$ be a tree decomposition of $H$ of width $tw(H)=\max\{|X|:X \in \mathcal{X}\}$. We construct a tree decomposition $\mathcal{T}_G$ of $G$ from $\mathcal{T}_H$ as follows. First, we add $a$ and $b$ to each node $X \in \mathcal{X}$. Next, we add a new leaf $X' = \{a,b,a',b'\}$ and connect it to any node. Finally, for each $j \in \{1,\ldots,n\}$, we create two new nodes $X_{j(1)}=\{j,j(1), a, b\}$ and $X_{j(2)}=\{j, j(1),j(2), a\}$. We add $X_{j(1)}$ as a leaf node connected to some node $X \in \mathcal{X}$ with $\{j,a,b\} \subseteq X$ (such a node must exist since $\mathcal{T}_H$ is a tree decomposition of $H$), and connect $X_{j(2)}$ to $X_{j(1)}$. It is not difficult to see $\mathcal{T}_G$ is indeed a tree decomposition of $G$, and that it has width $tw(H)+2$ since $H$ has at least one edge.

Let $\mathcal{T}_H= (\mathcal{X},F)$ be a path decomposition of $H$ of width $pw(H)=\max\{|X|:X \in \mathcal{X}\}$. We construct a path decomposition $\mathcal{T}_G$ of $G$ from $\mathcal{T}_H$ as follows. First, we add $a$ and $b$ to each node $X \in \mathcal{X}$. Next, we add a new leaf $X' = \{a,b,a',b'\}$ to either end of $\mathcal{T}_H$. Finally, for each $j \in \{1,\ldots,n\}$, let $X_j \in \mathcal{X}$ be any node such that $j \in X$. Such a node must exist since $\mathcal{T}_H$ is a path decomposition of $H$. Duplicate this node, call it $X_j'$ and make $X_j'$ a neighbor of $X_j$ on the path. We now create a new node $X_{j'}= X_j \cup \{j(1), j(2)\}$ and insert it on the path between $X_j$ and $X_j'$. Note that $a,b \in X_j$ by construction. It is now not difficult to see $\mathcal{T}_G$ is indeed a path decomposition of $G$, and that it has width $pw(H)+4$.
\qed
\end{proof}

\begin{lemma}
\label{lem:ForwardDirectionCmax}%
If there is a solution $\chi': \{1,\ldots,n\} \to \{1,\ldots,k\}$ to $(H,k,\chi)$, then there is a feasible schedule~$C$ for the constructed \threefield{P_n}{conc}{C_{\max}} instance with $C_{\max} \leq k$. 
\end{lemma}
\begin{proof}
It is easy to verify that the following schedule $C$ is feasible and has makespan at most $k$: set
$C_j = \chi'(j)$ for all $j \in \{1,\ldots,n\}$;
 $C_a = 1$, $C_b=k$, $C_{a'}=k$, and $C_{b'}=k-1$;
$C_{j(1)} = \chi'(j)-1$ for all $j \in V_0$ with $\chi'(j) > 1$;
 $C_{j(2)} = k$ for all $j \in V_0$ with $\chi'(j) < k$.
\qed
\end{proof}

\begin{lemma}
\label{lem:BackwardDirectionCmax}
If there is a feasible schedule~$C$ for the constructed \threefield{P_n}{conc}{C_{\max}} instance with $C_{\max} \leq k$, then there is a solution $\chi': \{1,\ldots,n\} \to \{1,\ldots,k\}$ to $(H,k,\chi)$.   
\end{lemma}

\begin{proof}
Let $C$ be a feasible schedule for the constructed \threefield{P_n}{conc}{C_{\max}} instance with $C_{\max} \leq k$. Then, according to Lemma~\ref{lem:delimiters}, either $C_a=1$ and $C_b=k$, or $C_a=k$ and $C_b=1$. 

Let us begin with the case of $C_a=1$ and $C_b=k$. 
Define $\chi'(j) = C_j$.
Consider any $j\in V_0$, and suppose that $\chi(j) \in \{2,\ldots,k-1\}$. Then as jobs~$j$,~$j(1)$, and~$j(2)$ are mutually conflicting, we have $C_j \neq C_{j(1)} \neq C_{j(2)}$. Moreover, as the total processing time of these three jobs is~$k = C_{\max}$, one of them must complete at time $k$, and since both $j$ and $j(1)$ are in conflict with job~$b$, this can only be job $j(2)$. Thus, $C_{j(2)}=k$. Again, as the total processing time of all three jobs is $k$, one of $j$ and~$j(2)$ must complete right before job $j(2)$ starts, \emph{i.e.} at time $C_{j(2)}-p_{j(2)}=k-(k-\chi(j))=\chi(j)$. If job $j(1)$ completes at this time, then job~$j$ completes at time $C_{j(1)}-p_{j(1)}=\chi(j)-(\chi(j)-1)=1$, which contradicts the feasibility of $C$ since $j$ and $a$ are in conflict. Thus, it must be that $C_j = \chi(j)$.  

The cases where~$\chi(j)=1$ or~$\chi(j)=k$ are very similar. If~$\chi(j)=1$, then $C_{j(2)}=k$, or we would have~$C_j=k$, which contradicts the feasibility of $C$, as $j$ and $b$ are in conflict. Thus, job~$j$ completes at time~$C_{j(2)} - p_{j(2)} = k - (k-1) = 1= \chi(j)$. If~$\chi(j)=k$, then $C_{j(1)}=C_{j(1)}=k-1$ and $j$ completes at time $C_{j(1)} - p_j = k-1 +1 = k= \chi(j)$. It therefore follows that if $C_a=1$ and $C_b=k$, then $C_j = \chi(j)$ for all~$j \in V_0$. Thus, since $C$ is feasible, the restriction $\chi'$ of $C$ onto $\{1,\ldots,n\}$ is a proper coloring of $H$, and so~$\chi'$ is a solution for $(H,k,\chi)$.

Next suppose that $C_a=k$ and $C_b=1$. Then following the same arguments as above, we get that $C_j = k-\chi(j)+1$ for all~$j \in V_0$. But then the coloring $\chi'$ with $\chi'(j)=k-C_j+1$ for all $j \in \{1,\ldots,n\}$ is also a proper coloring of $H$. Moreover, $\chi'$ uses at most $k$ colors, and $\chi'(j)=\chi(j)$ for all $j \in V_0$. Thus, $\chi'$ is a solution for $(H,k,\chi)$.
\qed
\end{proof}

The reduction can be carried out in polynomial time and an additional $O(\log n)$ working space. Correctness of the reduction is given by Lemma~\ref{lem:ForwardDirectionCmax} and Lemma~\ref{lem:BackwardDirectionCmax}. Moreover, by Lemma~\ref{lem:TreewidthCmax} we have $tw(G) = O(tw(H))$. Thus, in total, we have described a parameterized logspace reduction from the XALP-hard \textsc{Precoloring Extension} problem to \threefield{P_n}{conc}{C_{\max}}. Since XALP-membership follows from Corollary~\ref{corollary:inXALP}, we obtain the following theorem: 

\begin{theorem}
\threefield{P_n}{conc}{C_{\max}}, parameterized by treewidth, is XALP-complete.     
\end{theorem}

\subsection{Total completion time}
Let $(H,k,\chi)$ be a given instance of \textsc{Precoloring Extension} with $H=(\{1,\ldots,n\},F)$, $k$ an integer, and $\chi: V_0 \to \{1,\ldots,k\}$ a given precoloring of $H$ defined on a subset of vertices $V_0 \subseteq V$. 
We may assume that $H$ has at least one edge, and thus $n \geq 2$, or the instance is trivial. 
On a high level the main idea of our reduction is similar to the one for \threefield{P_n}{conc}{\sum C_{\max}}: We initially start with $H$ as our conflict graph for the set of jobs $\{1,\ldots,n\}$. We then add auxiliary jobs to each job $j \in \{1,\ldots,n\}$ whose goal is to ensure that in an optimal schedule $C$, jobs $j \in V_0$ complete precisely at time $\chi(j)$, while jobs in $\{1,\ldots,n\} \setminus V_0$ complete anywhere in $\{1,\ldots,k\}$. However, the technical details here are a bit more involved in comparison to the \threefield{P_n}{conc}{\sum C_{\max}} problem. 

We first set $X=nk+1$, and $G=H$ with $p_j=1$ for all $j \in \{1,\ldots,n\}$. We then construct a set of \emph{upper bound jobs} for each $j \in \{1,\ldots,n\}$. Define $u_j = \chi(j)$ for each $j \in V_0$, and $u_j = k$ for each $j \in \{1,\ldots,n\} \setminus V_0$. Then $u_j$ is an upper bound on the color assigned to vertex $j$ in a solution $\chi'$ for $(H,k,\chi)$; that is, $\chi'(j) \leq u_j$ for any such $\chi'$. For each $j \in \{1,\ldots,n\}$, we construct (see Figure~\ref{fig:SumCj}): 
\begin{itemize}
\item jobs $j(1),\ldots,j(X)$, each having a processing time of $X-u_j$. These are the \emph{primary} upper bound jobs of $j$. Add the edges $\{j,j(1)\}, \ldots,\{j,j(X)\}$ to $G$. 
\item jobs $j(i,1),\ldots,j(i,X)$ for each $i \in \{1,\ldots,X\}$, each with processing time~$u_j$. Add all edges $\{j(i),j(i,1)\},\ldots,\{j(i),j(i,X)\}$, for each $i \in \{1,\ldots,X\}$, to $G$. We collectively call all these jobs the \emph{secondary} upper bound jobs of $j$. Note that the secondary upper bound jobs of $j$ are not in conflict with $j$ itself. 
\end{itemize}
Next, we construct \emph{lower bound jobs} for each $j \in \{1,\ldots,n\}$. Define $\ell_j = \chi(j)$ for each $j \in V_0$, and $\ell_j = 1$ for each $j \in \{1,\ldots,n\} \setminus V_0$. Then $\ell_j$ is a lower bound on the color assigned to vertex $j$ in a solution $\chi'$ for $(H,k,\chi)$; that is, $\chi'(j) \geq \ell_j$ for any such $\chi'$. For each $j \in V_0$ with $\chi(j) > 1$, we construct (see Figure~\ref{fig:SumCj}): 
\begin{itemize}
\item jobs $j^*(1),\ldots,j^*(X)$, each of which has processing time of $\ell_j -1$ and is in conflict only with $j$. We add all edges $\{j,j^*(1)\},\ldots,\{j,j^*(X)\}$ to $G$.
\end{itemize}
This completes our construction of the jobs for the constructed \threefield{P_n}{conc}{\sum C_j} instance. Note that $tw(G)=tw(H)$, as one only needs to append nodes of size~$2$ to a tree decomposition of $H$ to obtain a tree decomposition of $G$. Similarly, $pw(G) \leq pw(H)+2$, as one only needs to add copies of existing nodes with two additional vertices. Finally, we set the required total completion time of all jobs to be 
$K = \sum^n_{j=1} \Big[ (1+u_j) \cdot X^2 + (\ell_j-1) \cdot X \Big] + nk$.

\begin{figure}[t!]
\centering
\includegraphics[width=\linewidth]{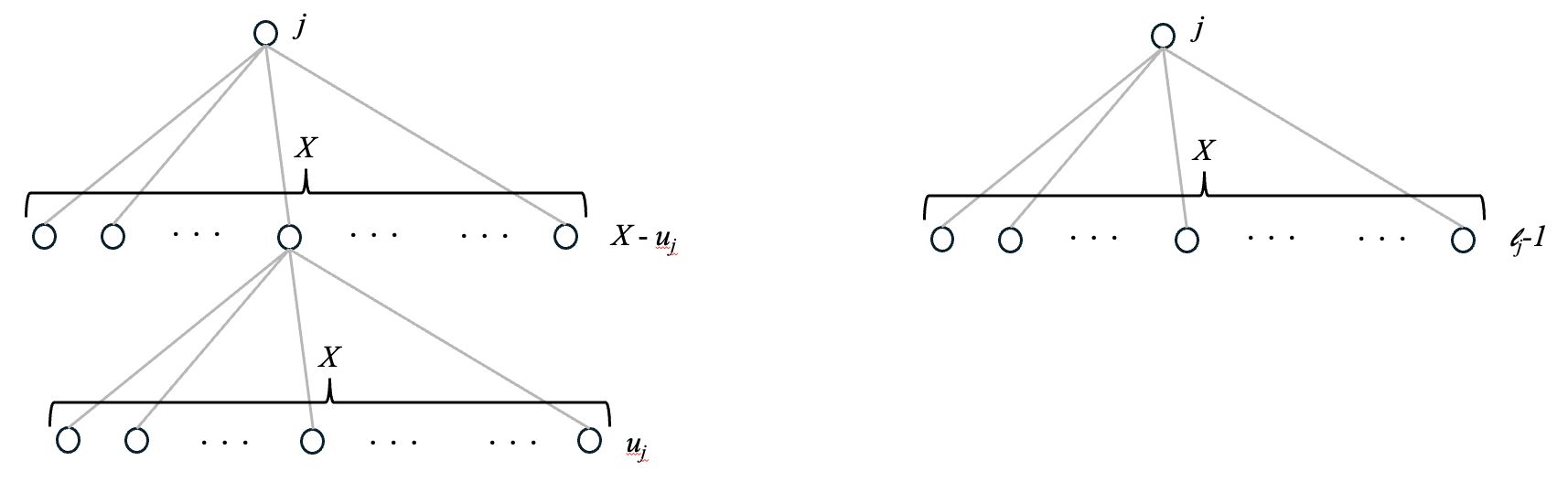}
\caption{The conflict graph of the upper bound (left) and lower bound jobs (right) of some job~$j$. The processing times are depicted to the right of each job. \label{fig:SumCj}}
\end{figure}

\begin{lemma}
\label{lem:ForwardDirectionSumCj}%
If there is a solution $\chi': \{1,\ldots,n\} \to \{1,\ldots,k\}$ to $(H,k,\chi)$, then there is a feasible schedule~$C$ for the constructed \threefield{P_n}{conc}{\sum C_j} instance with total completion time at most $K$.
\end{lemma}

\begin{proof}
Define a schedule $C$ as follows: 
\begin{itemize}
\item $C(j) = \chi'(j)$ for all $j \in \{1,\ldots,n\}$.
\item $C(j(i))=X$ for all $j \in \{1,\ldots,n\}$ and $i \in \{1,\ldots,X\}$.
\item $C(j(i,i_0))=u_j$ for all $j \in \{1,\ldots,n\}$ and $i,i_0 \in \{1,\ldots,X\}$.
\item $C(j^*(i))= \ell_j-1$ for all $j \in V_0$ such that $\chi(j) > 1$ and $i \in \{1,\ldots,X\}$.
\end{itemize}
Then $C$ is feasible, as $\chi'$ is a proper coloring on $H$ with $\ell_j < \chi'(j) \leq u_j$ for all $j \in \{1,\ldots,n\}$. Moreover, for every $j \in \{1,\ldots,n\}$, the total completion time of all upper bound jobs of $j$ is~$X^2\cdot u_j + X \cdot X = (1+u_j)\cdot X^2$, and the total completion time of all lower bound jobs of $j$ is~$(\ell_j-1)\cdot X$. As $\sum^n_{j=1} C(j) = \sum^n_{j=1} \chi'(j) \leq nk$, it follows that the total completion time of all jobs in $C$ is at most $K$. 
\qed
\end{proof}

We now prove the converse direction, where the idea is that if $C(j) < \ell_j$ or $C(j) > u_j$, then the total completion time of the auxiliary jobs is too large.

\begin{lemma}
\label{lem:UpperBoundJobs}%
Let $C$ be a minimal feasible schedule for the constructed \threefield{P_n}{conc}{\sum C_j} instance. Consider some $j \in \{1,\ldots,n\}$, and let~$J$ denote the set of upper bound jobs of~$j$. Then the total completion time of all jobs in~$J$ in $C$ is at least~$(u_j+1) \cdot X^2$. Moreover, if $C(j) \in \{u_j+1,\ldots,X\}$, then the total completion time of all jobs in~$J$ is at least~$(u_j+1) \cdot X^2 + X$.
\end{lemma}

\begin{proof}
Consider a primary upper bound job $j(i)$ for some $i \in \{1,\ldots,X\}$. If $C(j(i)) \geq X$, then by minimality of the schedule, all secondary upper bound jobs $j(i,1),\ldots,j(i,X)$ complete at the same time $u_j$. Hence, $C(j(i)) + \sum_{i_0=1}^{X} C(j(i,i_0)) \geq X(u_j+1)$. Note that if $C(j(i)) > X$, then this sum is strictly larger than $X(u_j+1)$. If $C(j(i)) < X$, then all secondary upper bound jobs $j(i,1),\ldots,j(i,X)$ conflict with $j(i)$ in the schedule and, by minimality of the schedule, complete at the same time $C(j(i,1)) = \cdots = C(j(i,X)) \geq X$. Hence, since $n \geq 2$, $C(j(i)) + \sum_{i_0=1}^{X} C(j(i,i_0)) \geq (X-u_j) + X\cdot X > X \cdot X > X \cdot (u_j+2) = X(u_j+1) + X$. It follows immediately that the total completion time of all jobs in~$J$ in $C$ is at least~$(u_j+1) \cdot X^2$.

Suppose that $C(j) \in \{u_j+1,\ldots,X-u_j\}$. Then all primary upper bound jobs conflict with $j$ in the schedule and complete at time larger than $X$. From the above, the total completion time of all jobs in~$J$ is at least~$(u_j+1) \cdot X^2 + X$.

Suppose that $C(j) \in \{X-u_j+1,\ldots,X\}$. Then either all primary upper bound jobs complete at time larger than $X$ or there is a primary upper bound job that completes before time $C(j)-1 < X$ (note that $C(j)-1 \geq X-u_j$). In the former case, as before, the total completion time of all jobs in~$J$ is at least~$(u_j+1) \cdot X^2 + X$. If there is a primary upper bound job (say $j(i)$) that completes before time $C(j)-1 < X$, then by the above, the total completion time of all jobs in~$J$ is at least~$(u_j+1) \cdot X^2 + X$ as well.
\qed
\end{proof}

\begin{lemma}
\label{lem:LowerBoundJobs}
Let $C$ be a minimal feasible schedule for the constructed \threefield{P_n}{conc}{\sum C_j} instance. Consider some $j \in V_0$ with $\chi(j) >1$, and let $J$ denote the set of lower bound jobs of $j$. Then the total completion time of all jobs in $J$ in $C$ is at least $(\ell_j-1) \cdot X$. Moreover, if $C(j) < \ell_j$, then the total completion time of all jobs in~$J$ is at least $\ell_j \cdot X$.
\end{lemma}

\begin{proof}
The total completion time of $J$ is at least $(\ell_j-1) \cdot X$ because this is the total processing time of $J$. If $C(j) < \ell_j$, then the completion time of any job in $J$ is at least $\ell_j$, as all jobs in $J$ are conflicting with $j$ and have processing time $\ell_j-1$. Since $|J|=X$, the total completion time of $J$ is at least $\ell_j \cdot X$.   
\qed
\end{proof}

Combining Lemma~\ref{lem:UpperBoundJobs} with Lemma~\ref{lem:LowerBoundJobs} we obtain the following: 

\begin{lemma}
\label{lem:BackwardDirectionSumCj}%
If there is a minimal feasible schedule~$C$ for the constructed \threefield{P_n}{conc}{\sum C_j} instance with $\sum C_j \leq K$, then there is a solution $\chi': \{1,\ldots,n\} \to \{1,\ldots,k\}$ for $(H,k,\chi)$.   
\end{lemma}

\begin{proof}
Let $C$ be a minimal feasible schedule for the constructed \threefield{P_n}{conc}{\sum C_j} instance with total completion time at most $K$. By Lemma~\ref{lem:UpperBoundJobs} and Lemma~\ref{lem:LowerBoundJobs}, the total completion time of all jobs other than $\{1,\ldots,n\}$ in $C$ is at least 
$
K' = \sum^n_{j=1} \Big[ (1+u_j) \cdot X^2 + (\ell_j-1) \cdot X \Big]
$.
Since $K-K' = nk < X$, we have that $C(j) < X$. Therefore, $C(j) \leq u_j$ for all $j \in \{1,\ldots,n\}$ by Lemma~\ref{lem:UpperBoundJobs}. Also, $C(j) \geq \ell_j$ for all $j \in V_0$ by Lemma~\ref{lem:LowerBoundJobs}. It follows that the restriction of $C$ onto $\{1,\ldots,n\}$ is a proper extension of the coloring $\chi$ onto $\{1,\ldots,n\}$.
\qed
\end{proof}

The reduction above can be carried out in polynomial time and an additional $O(\log n)$ working space. Moreover, $tw(G) = tw(H)$. Correctness of the reduction is given by Lemma~\ref{lem:ForwardDirectionSumCj} and Lemma~\ref{lem:BackwardDirectionSumCj}. This yields a parameterized logspace reduction from the XALP-hard \textsc{Precoloring Extension} to \threefield{P_n}{conc}{\sum C_j}. 
Membership in XALP follows again from Corollary~\ref{corollary:inXALP}, and so: 

\begin{theorem}
\threefield{P_n}{conc}{\sum C_j} parameterized by treewidth is XALP-complete.
\end{theorem}

\section{Non-Zero Release Times}
\label{section:non-zeroreleasetimes}

We now consider the case where each job $j$ has an integer release time $r_j \geq 0$ and in a feasible schedule $C$ we require that $C(j) \geq r_j + p_j$ for all jobs~$j$. We show that the complexity of the concurrency scheduling problems considered above substantially differs from the $r_j=0$ case. As all problems are XALP-complete for arbitrary processing times, we only consider the unit processing time case.

We begin with makespan objective. By a simple reduction to the \threefield{P_n}{conc,p_j=1}{L_{\max}} problem, by reversing the order of schedules, we obtain:

\begin{theorem}
\label{thm:UnitProcHard}%
\threefield{P_n}{conc,p_j=1,r_j \geq 0}{C_{\max}} is solvable in $O^*((tw \cdot \log n+1)^{tw+1})$-time.
\end{theorem}

\begin{proof}
Let $G=(\{1,\ldots,n\},E)$ be a given \threefield{P_n}{conc,p_j=1,r_j \geq 0}{C_{\max}} instance with $tw(G)=tw$, and let $k$ be the given makespan bound. Construct an instance of \threefield{P_n}{conc,p_j=1}{L_{\max}} with job set $\{1,\ldots,n\}$ and the same conflict graph~$G$, by setting $d_j=k-r_j$. It is not difficult to see that if $C$ is a feasible schedule for the constructed \threefield{P_n}{conc,p_j=1}{L_{\max}} instance with $L_{\max}=0$, then the schedule $C'$ obtained by setting $C'(j)=k-C(j)+1$ for each job~$j$ is feasible for the given \threefield{P_n}{conc,p_j=1,r_j \geq 0}{C_{\max}} instance, as $r_j +1 \leq C'(j)  \leq k$ for each job~$j$. 
Conversely, if $C'$ is a feasible schedule for the given \threefield{P_n}{conc,p_j=1,r_j \geq 0}{C_{\max}} instance, then the schedule $C$ we obtain by setting $C(j) = k-C'(j)+1$ is feasible for the constructed \threefield{P_n}{conc,p_j=1}{L_{\max}} instance with $L_{\max}=0$. 
Thus, using Theorem~\ref{thm:UnitProc}, we get an algorithm for \threefield{P_n}{conc,p_j=1,r_j \geq 0}{C_{\max}} running in $O^*((tw \cdot \log n+1)^{tw+1})$-time. 
\qed
\end{proof}

We next consider the $L_{\max}$ and $WC_{\max}$ objectives. For both of these objectives we previously obtained FPT algorithms when all jobs are released at time~$0$. However, when this is not true, the problems become XALP-complete:

\begin{theorem}
\threefield{P_n}{conc, p_j =1, r_j \geq 0}{L_{\max}} and \threefield{P_n}{conc, p_j =1, r_j \geq 0}{WC_{\max}} are XALP-complete parameterized by treewidth.  
\label{theorem:x1}
\end{theorem}

\begin{proof}
Membership follows again from Corollary~\ref{corollary:inXALP}.
For the hardness of both problems we present reductions from \textsc{Precoloring Extension}. Let $(H,k,C)$ be a given instance of \textsc{Precoloring Extension} with $H=(\{1,\ldots,n\},F)$, $k$ an integer, and $\chi: V_0 \to \{1,\ldots,k\}$ a given precoloring of $H$ defined on some $V_0 \subseteq V$.

For \threefield{P_n}{conc, p_j =1, r_j \geq 0}{L_{\max}}, we construct an instance $G=(\{1,\ldots,n\},E)$ by setting $G=H$, and assigning each job~$j \in V_0$ a release time of~$r_j = \chi(j)$ and due date~$d_j=\chi(j)+1$. All other jobs~$j \notin V_0$ are assigned a release date of~$r_j=0$ and due date~$d_j=k$. Clearly, in any feasible schedule~$C$ with~$L_{\max} \leq 0$ we have $C(j)=\chi(j)$ for all $j \in V_0$, and $1 \leq C(j) \leq k$ for all $j \notin V_0$. Thus, since $G=H$, any feasible schedule with $L_{\max} \leq 0$ corresponds to a solution for the \textsc{Precoloring Extension} instance, and vice versa.

For \threefield{P_n}{conc, p_j =1, r_j \geq 0}{WC_{\max}}, we construct an instance $G=(\{1,\ldots,n\},E)$ by setting $G=H$, and assigning each job~$j \in V_0$ a release time of~$r_j = \chi(j)$ and a weight of~$w_j=\lfloor k/(\chi(j)+1) \rfloor$. All other jobs $j \notin V_0$ are assigned a release date of~$r_j=0$ and weight~$w_j=1$. It is easy to see that in any feasible schedule $C$ with $WC_{\max} \leq k$ we have $C(j)=\chi(j)$ for all $j \in V_0$, and $1 \leq C(j) \leq k$ for all~$j \notin V_0$. Thus, since $G=H$, any feasible schedule with $WC_{\max} \leq k$ corresponds to a solution for the \textsc{Precoloring Extension} instance, and vice versa.

As both constructions correspond to logspace parameterized reductions, the theorem follows.
\qed
\end{proof}

Finally, we consider the $\sum C_j$ objective.

\begin{theorem}
\threefield{P_n}{conc, p_j =1, r_j \geq 0}{\sum C_j} is XALP-complete parameterized by treewidth.  
\label{theorem:x2}
\end{theorem}
\begin{proof}
The proof is via a reduction from \textsc{Precoloring Extension}. Let $(H,k,\chi)$ be a given instance of \textsc{Precoloring Extension} with $H=(\{1,\ldots,n\},F)$, $k$ an integer, and $\chi: V_0 \to \{1,\ldots,k\}$ a given precoloring of $H$ defined on some $V_0 \subseteq V$. 
We may assume that $H$ has at least one edge.
Define $u_j = \chi(j)$ for each $j \in V_0$, and $u_j = k$ for each $j \in \{1,\ldots,n\} \setminus V_0$. Furthermore, define $\ell_j = \chi(j)-1$ for each $j \in V_0$, and $\ell_j = 0$ for each $j \in \{1,\ldots,n\} \setminus V_0$. Then $\chi':\{1,\ldots,k\}$ is a solution for $(H,k,C)$ if and only if it is a proper coloring of $H$ with $\ell_j < \chi'(j) \leq u_j$ for all~$j \in \{1,\ldots,n\}$. Let $X=nk+1$.

We construct an instance of \threefield{P_n}{conc, p_j =1, r_j \geq 0}{\sum C_j} by initially starting with the set of $n$ jobs $\{1,\ldots,n\}$, corresponding to the vertices of~$H$, and a conflict graph~$G=H$. We set the release time of each job $j \in \{1,\ldots,n\}$ to~$r_j=\ell_j$. We then add, for each $j \in \{1,\ldots,n\}$, a set of upper bound jobs. For $i \in \{u_j,\ldots,X-2\}$, add jobs $j(i,1),\ldots,j(i,X)$ each released at time $i$. Finally, we add the edges $\{j,j(i,i_0)\}$ to $G$ for each $j \in \{1,\ldots,n\}$, $i \in \{u_j,\ldots,X\}$, and $i_0 \in \{1,\ldots,X\}$. We set the required total completion time of all jobs to be 
$$
K = nk + X \cdot \sum^n_{j=1} \sum_{i=u_j}^{X-2} (i+1).
$$
Note that $tw(G) = tw(H)$, as one only needs to append nodes of size~$2$ to a tree decomposition of $H$ to obtain a tree decomposition of $G$. (Also note that $pw(G) \leq pw(H)+1$.) 
We claim that there exists a feasible minimal schedule $C$ for the constructed set of jobs with total completion time at most $K$ if and only if there exists a solution $\chi'$ for $(H,k,\chi)$.

Suppose there exists a solution $\chi'$ for $(H,k,\chi)$. Then it is easy to verify that the schedule $C$ with, for all $j \in \{1,\ldots,n\}$, $C(j)=\chi'(j)$ and $C(j(i,i_0))=i+1$ for all $i \in \{u_j,\ldots,X-2\}$ and $i_0 \in \{1,\ldots,X\}$, is feasible and has total completion time at most $K$.

Suppose there exists a minimal feasible schedule $C$ for the constructed instance with total completion time at most $K$. Note that for all $j \in \{1,\ldots,n\}$ and $i \in \{u_j,\ldots,X-2\}$, the completion time of the upper bound jobs $j(i,1),\ldots,j(i,X)$ is at least $i+1$ by construction. This accounts for a sum of completion times of at least $X \cdot \sum^n_{j=1} \sum_{i=u_j}^{X-2} (i+1)$. Suppose that for some $j \in \{1,\ldots,n\}$, $i \in \{u_j,\ldots,X-2\}$, and $i_0 \in \{1,\ldots,X\}$ it holds that $C(j(i,i_0)) > i+1$. By minimality of the schedule, all upper bound jobs $j(i,1),\ldots,j(i,X)$ complete at the same time, as they are only in conflict with job $j$. Hence, the total completion time of all jobs is at least $X + X \cdot \sum^n_{j=1} \sum_{i=u_j}^{X-2} (i+1) > K$, a contradiction. Hence, $C(j(i,i_0)) = i+1$ for all $j \in \{1,\ldots,n\}$, $i \in \{u_j,\ldots,X-2\}$, and $i_0 \in \{1,\ldots,X\}$. 
It follows that $C(j) < X$ for any $j \in \{1,\ldots,n\}$, or we exceed the budget $K$. Since $j$ is in conflict with all its upper bound jobs, it follows that $C(j) \leq u_j$. Moreover, by construction, $C(j) > \ell_j$. Now create a coloring $\chi'$ by setting $\chi'(j) = C(j)$ for all $j \in \{1,\ldots,n\}$. This is a proper coloring with $\ell_j < \chi'(j) \leq u_j$ for all~$j \in \{1,\ldots,n\}$, as required.

Since this construction corresponds to a logspace parameterized reduction, the theorem follows.
\qed
\end{proof}

\section{Discussion}
In this paper, we investigated the complexity of scheduling problems with concurrency constraints, where the constraint graph has bounded treewidth or bounded pathwidth.
We considered different variants of the problem, with four different objective functions, uniform or different release times, and unit or variable job lengths. 
For these cases, we obtained a dichotomy concerning the fixed-parameter tractability of the problems: some variants are fixed-parameter tractable, while the remaining cases are XALP-complete. 
Several of the positive results make use of the fact that the Grundy number of a graph of treewidth $k$ is bounded by $O(k \log n)$, and of a graph of pathwith $k$ is
bounded by $O(k)$.


We may also consider further aspects of scheduling problems. First, if we consider preemption, then $p_j \geq 1$ but the times at which a job is processed is not necessarily a contiguous interval. The problem \threefield{P_n}{conc,pmtn}{C_{\max}} is known as {\sc Preemptive Multicoloring}, {\sc Multichromatic Number}, {\sc Weighted Coloring}, or {\sc Minimum Integer Weighted Coloring}~\cite{BalasX91,BalasX92,GLS1988,Hoang94,ItoNZ02,ItoNZ03,McDiarmidR00,ZhouN03}. We can reduce {\sc Preemptive Multicoloring} to {\sc Coloring} by replacing each job $j$ by $p_j$ jobs that are pairwise in conflict. That is, we replace each vertex by $p_j$ (true) twins. 
Note that this step can increase the pathwidth and treewidth
of the conflict graph by a multiplicative factor of $\max p_j$.
Thus, we can transfer positive results for the non-preemptive case with $p_j=1$ to positive results for the preemptive case with
$p_j\geq 1$, when taking $\max p_j$ as additional parameter. It is
interesting to investigate what happens when the $p_j$ are not
parameter bounded.
We may also consider \threefield{P_n}{conc,pmtn}{\sum C_j}, known as \textsc{Preemptive Multichromatic Sum}~\cite{Bar-NoyHKSS00,HalldorssonK02}.

Second, we can consider a bounded number $m$ of machines. If $p_j = 1$, then \threefield{P_m}{conc}{C_{\max}} is equivalent to the well-studied {\sc $m$-Bounded Coloring} problem and has also been called {\sc Mutual Exclusion Scheduling} and {\sc $\chi_m$-coloring}; see e.g.~\cite{Alon83,BakerC96,BodlaenderF05,BodlaenderJ93,BodlaenderJ95,HansenHK93,KallerGS95,Lonc91}. In particular, for fixed $m$, this problem is fixed-parameter tractable parameterized by the treewidth of the conflict graph~\cite{KallerGS95}, while for variable $m$, it is W[1]-hard~\cite{FellowsFLRSST11} and in XP~\cite{BodlaenderF05} parameterized by the treewidth of the conflict graph plus the makespan. This problem is also strongly related to the {\sc Equitable Coloring} problem, where the goal is to color a graph such that the size of any two color classes differs by at most~$1$.

\subsubsection*{Acknowledgments}
The first author wants to thank Meike Hatzel, Pascal Gollin,
Matja\v{z} Krnc, and Martin Milani\v{c} for discussions on \textsc{Chromatic Sum} 
and related problems. An extended abstract of this paper appears in the proceedings of the 51st International Workshop on Graph-Theoretic Concepts in Computer Science, WG 2025, published in the Springer Lecture Notes in Computer Science series.

\bibliographystyle{abbrvurl}
\bibliography{references}

\end{document}